%% file: main.tex
\documentclass[runningheads]{llncs/llncs}

\input{header}

\begin{document}
\title{Regular Model Checking with Regular Relations}
%
%\titlerunning{Abbreviated paper title}
% If the paper title is too long for the running head, you can set
% an abbreviated paper title here
%
\author{Vrunda Dave\inst{1} \and
Taylor Dohmen\inst{2} \and\\
Shankara Narayanan Krishna
\inst{1} \and
Ashutosh Trivedi\inst{2}}
\authorrunning{Dave, Dohmen, Krishna, and Trivedi}
% First names are abbreviated in the running head.
% If there are more than two authors, 'et al.' is used.
%
\institute{IIT Bombay, Mumbai, India \\
\email{\{vrunda,krishnas\}@cse.iitb.ac.in} \and
Univeristy of Colorado, Boulder, USA \\
\email{\{taylor.dohmen,ashutosh.trivedi\}@colorado.edu}}

\maketitle

\begin{abstract}
  \input{abstract}

\end{abstract}

\section{Introduction}
\label{sec:introduction}
\input{introduction}

\section{Regular Relations for Infinite Strings}
\label{sec:wnsst}
\input{wnsst}

\section{MSO-Definable Regular Model Checking} % with Regular Relations}
\label{sec:dec}
\input{decision}

\begin{figure}
  \centering
  % \vspace{-1.5em}
  % \scalebox{.8}{
    \begin{tikzpicture}[node distance = 4cm, >= stealth]
      \node[state,initial]   (1) {$s_0$};
      \node[state,accepting] (2)[right of = 1] {$s_1$};
      \node[state,accepting] (3)[right of = 2] {$s_2$};
          \path[->] (1) edge  node[above] {$1  \begin{cases} x := \varepsilon\\ y
              := 0 \\z = 1 \end{cases}$}  (2);
              \path[->] (2) edge[loop above] node[above] {$1 \begin{cases} x:=x1
                  \\ y:= y0 \\z=z \end{cases}$}  (2);
              \path[->] (2) edge node[above] {$0 \begin{cases} x := xy1 \\
                  y:= \varepsilon \\ z=z \end{cases}$} (3);
              \path[->] (3) edge[loop above] node {$0 \begin{cases} x := x0\\
              		y:= \varepsilon \\ z=z \end{cases}$}  (3);
  \end{tikzpicture}
  % }
  \caption{An $\wsst{}$ squaring a number with binary expansion of the form $1^n0^\omega$. The output at $s_1$ and $s_2$ is $x$. Notice that this function can not be expressed as a \gsm{}.}
  \label{fig:sst-square}
  \vspace{-2em}
\end{figure}

\section{Conclusion}
\input{conclusion}

\bibliographystyle{llncs/splncs04}
\bibliography{rmc}

\appendix

\section{Formulae for Example 2}
\label{sec:examples}
\input{formulae}

\section{Undecidability of Unbounded Regular Model Checking}
\label{sec:undec_proof}
\input{undec}

\end{document}

%% file: header.tex
\usepackage{url}
\usepackage{graphicx}
\usepackage{amssymb,amsmath,amsfonts}
\usepackage{mathtools}
\usepackage[ruled,linesnumbered]{algorithm2e}
\usepackage[noend]{algpseudocode}
\usepackage{bookmark}
\usepackage{wrapfig}
\usepackage{thmtools}
\usepackage{thm-restate}
\usepackage{tikz}
\usetikzlibrary{automata,positioning,decorations.markings,arrows,intersections, calc,shapes,fit}

\tikzset{
  >=latex,
  node distance=1cm
}
\tikzstyle{every state}=[draw=black, very thick, fill = black!10!white, minimum size = 2mm]

\usepackage{framed}
\usepackage{array}
\usepackage{booktabs}
\usepackage[textwidth=1.5cm,textsize=scriptsize,backgroundcolor=blue!10!white]{todonotes}
\usepackage{color}
\usepackage{cleveref}
\usepackage{hyperref}
\usepackage{enumitem}

\newcommand{\tn}[1]{\textnormal{#1}}

\newcommand{\set}[1]{\left\{ #1 \right\}}

\newcommand{\seq}[1]{\langle #1 \rangle}
\def\rmdef{\stackrel{\mbox{\rm {\tiny def}}}{=}} % equals definition

\newcommand{\inter}[1]{[\![#1]\!]}
\newcommand{\sem}[1]{[\![#1]\!]}
\newcommand{\dom}[1]{\mathsf{dom}(#1)}
\newcommand{\im}[1]{\mathsf{im}(#1)}
\newcommand{\domain}{\mathsf{dom}}
\newcommand{\mso}{\tn{MSO}}

\newcommand{\dmsot}{\tn{DMSOT}}
\newcommand{\nmsot}{\tn{NMSOT}}
\newcommand{\wnmsot}{\tn{$\omega$-NMSOT}}
\newcommand{\wmsot}{\tn{$\omega$-MSOT}}
\newcommand{\wdmsot}{\tn{$\omega$-DMSOT}}

\newcommand{\sst}{\tn{SST}} 

\newcommand{\dmt}{\tn{DMT}} 
\newcommand{\nmt}{\tn{NMT}} 
\newcommand{\nbt}{\tn{NBT}} 
\newcommand{\dbt}{\tn{DBT}}

\newcommand{\nba}{\tn{NBA}}
\newcommand{\dma}{\tn{DMA}}
\newcommand{\nma}{\tn{NMA}}

\newcommand{\nsst}{\tn{NSST}}
\newcommand{\wnsst}{\tn{$\omega$-NSST}}
\newcommand{\wdsst}{\tn{$\omega$-DSST}}
\newcommand{\wsst}{\tn{$\omega$-SST}}

\newcommand{\gsm}{\tn{GSM}}

\newcommand{\tgsm}{\tn{2GSM}}
\newcommand{\ntgsm}{\tn{N2GSM}}

\newcommand{\pspace}{\textsc{Pspace}}
\newcommand{\nlogspace}{\textsc{NLogSpace}}

\newcommand{\npspace}{\textsc{NPspace}}

\newcommand{\wstring}{$\omega$-string}
\newcommand{\wlang}{$\omega$-language}
\newcommand{\wreg}{$\omega$-regular}
\newcommand{\Inf}[1]{\mathsf{Inf}(#1)}

\newcommand{\acc}{\mathsf{Acc}}
\newcommand{\rl}{\tn{RL}}

\newcommand{\Ff}{\mathbb{F}}

\newcommand{\M}{\mathcal{M}}

\renewcommand{\S}{\mathcal{S}}

\newcommand{\fF}{\mathbb{F}}

\newcommand{\nats}{\mathbb{N}}

\newcommand{\runs}[2]{\mathsf{Runs}_{#1} (#2)}

\newcommand{\val}[1]{\mathsf{val}_{#1}}

\newcommand{\out}{f}

\renewcommand{\vec}[1]{\mathbf{#1}}

\usepackage{environ}

\newcommand{\repeattheorem}[1]{%
  \begingroup
  \renewcommand{\thetheorem}{\ref{#1}}%
  \expandafter\expandafter\expandafter\theorem
  \csname reptheorem@#1\endcsname
  \endtheorem
  \endgroup
}

\NewEnviron{reptheorem}[1]{%
  \global\expandafter\xdef\csname reptheorem@#1\endcsname{%
    \unexpanded\expandafter{\BODY}%
  }%
  \expandafter\theorem\BODY\unskip\label{#1}\endtheorem
}

\newcommand{\repeatlemma}[1]{%
  \begingroup
  \renewcommand{\thelemma}{\ref{#1}}%
  \expandafter\expandafter\expandafter\lemma
  \csname replemma@#1\endcsname
  \endlemma
  \endgroup
}

\NewEnviron{replemma}[1]{%
  \global\expandafter\xdef\csname replemma@#1\endcsname{%
    \unexpanded\expandafter{\BODY}%
  }%
  \expandafter\lemma\BODY\unskip\label{#1}\endlemma
}

\newcommand{\init}{\textsc{init}}
\newcommand{\bad}{\textsc{bad}}

\newcommand{\Aa}{\mathcal{A}}
\newcommand{\letter}{\mathsf{Let}}

%% file: abstract.tex
%%% Abstract goes here
Regular model checking is an exploration technique for infinite
state systems where state spaces are represented as regular languages and
transition relations are expressed using rational relations over infinite (or finite) strings.
We extend the regular model checking paradigm to permit the use of more powerful transition relations: the class of regular relations, of which the rational relations are a strict subset.
We use the language of monadic second-order logic (\mso{}) on infinite strings to specify such relations and adopt streaming string transducers (\sst{}s) as a suitable computational model.
We introduce nondeterministic \sst{}s over infinite strings (\wnsst{}s) and show that they precisely capture the relations definable in \mso{}.
We further explore theoretical properties of \wnsst{}s required to effectively
carry out regular model checking.
In particular, we establish that the regular type checking problem for \wnsst{}s is decidable in \pspace{}.
Since the post-image of a regular language under a regular relation may not be
regular (or even context-free), approaches that iteratively compute the image
can not be effectively carried out in this setting.
Instead, we utilize the fact that regular relations are closed under composition, which, together with our decidability result,
provides a foundation for regular model checking with regular relations.

%% file: introduction.tex
Regular model checking \cite{AbdullaJonssonNilssonSaksena04,AbdullaJonssonNillsondOrsoSaksena12,BouajjaniJonssonNilssonTouili00,BoigelotWolper98,KestenMalerMarcusPnueliShahar01} is a symbolic  exploration and verification technique where sets of configurations are
expressed as regular languages and transition relations are encoded as
rational relations \cite{Schutzenberger76,Sakarovitch09,LodingSpinrath19} in the form of generalized sequential machines.
A generalized sequential machine (\gsm{}) is essentially a finite state machine with output capability; on every transition an input symbol is read, the state changes, and a finite string is appended to an output string (see \Cref{fig:gsm-double}, for instance, where the label $\alpha / s$ indicates that the machine reads the symbol $\alpha$ and writes the string $s$ on any such transition).
While regular model checking is undecidable in general, a number of
approximation schemes and heuristics \cite{BouajjaniJonssonNilssonTouili00,JonssonNilsson00,Touili01,DamsLakhnechSteffen01,AbdullaJonssonNilssondOrso02,BoigelotLegayWolper03,BouajjaniHabermehlVojnar04,HabermehlVojnar05}
have made it a practical verification approach.
It has, for example, been applied to verify programs with unbounded data structures such as lists and stacks~\cite{BouajjaniJonssonNilssonTouili00,AbdullaJonssonNilssonSaksena04}.
Moreover, since infinite strings over a finite alphabet can be naturally interpreted as real numbers in the unit interval, regular
model checking over infinite strings provides a
framework \cite{BoigelotLegayWolper04,BoigelotWolper02,Legay12,BoigelotJodogneWolper05,BouajjaniLegayWolper05,LegayWolper10} to analyze
properties of dynamical systems.

\begin{wrapfigure}{r}{0.28\textwidth}
  \begin{center}
    \vspace{-2em}
    \scalebox{0.9}{
    \begin{tikzpicture}[node distance = 2.5cm, thick, >= stealth]
      \node[state,initial above] (1) {$s_0$};
      \node[state,accepting]          (2)[right of = 1] {$s_1$};
      \path[->] (1) edge[bend left]  node[above] {$\alpha/0\alpha$} (2);
      \path[->] (1) edge[bend right]  node[below] {$\alpha/00\alpha$} (2);
      \path[->] (2) edge[loop above] node[above] {$1/1$} (2);
      \path[->] (2) edge[loop below] node[below] {$0/0$}  (2);
    \end{tikzpicture}
    }
  \end{center}
  \vspace{-2em}
  \caption{\footnotesize A \gsm{} that shifts a string to the right by $1$ or $2$, or equivalently realizing division of the binary encoding of real numbers in $[0,1]$ by $2$ or $4$.
  }
  \label{fig:gsm-double}
  \vspace{-1.5em}
\end{wrapfigure}
This paper generalizes the regular model checking approach so that transition relations can be expressed using \emph{regular relations} over infinite strings.
We propose the computational model of
nondeterministic streaming string transducers on infinite strings (\wnsst{}),
and explore theoretical properties of \wnsst{}s required to effectively
carry
out regular model checking.

\noindent\textbf{Regular Relations.}
While rational relations are capable of modelling a rich set of transition systems, their limitations can be observed by noting their inability to express common transformations such as $\mathsf{copy} \rmdef w \mapsto ww$ and $\mathsf{reverse} \rmdef w \mapsto \overleftarrow{w}$, where
the string $\overleftarrow{w}$ is the reverse of the string $w$.
Courcelle \cite{Courcelle94,CourcelleEngelfriet12} initiated the use of monadic second-order logic (\mso{}) in defining deterministic and nondeterministic graph-to-graph transformations which are known to include
some non-rational transformations like \textsf{copy} and \textsf{reverse}.
Engelfriet and Hoogeboom~\cite{EngelfrietHoogeboom01} showed that  deterministic \mso{}-definable transformations (\dmsot{}) over finite strings coincide exactly with the transformations that can be realized by generalizations of \gsm{}s that can read inputs in two directions (\tgsm).
Furthermore, they showed that this correspondence does not extend to the set of nondeterministic \mso{}-definable transformations (\nmsot{}) and nondeterministic \tgsm{}s (\ntgsm{}).
\begin{wrapfigure}{l}{0.3\textwidth}
  \begin{center}
    \vspace{-2em}
    \scalebox{0.9}{
    \begin{tikzpicture}
      \node[state,initial,accepting,initial text=] (A) {$s_0$};
      \path[->] (A) edge[loop above] node {$a \begin{cases}x := ax\end{cases}$} (A);
      \path[->] (A) edge[loop below] node {$b \begin{cases}x := bx\end{cases}$} (A);
      \path[->] (A) edge node[above] {$x$} (1,0);
    \end{tikzpicture}
    }
  \end{center}
  \vspace{-2em}
  \caption{\footnotesize\sst{} implementing \textsf{reverse}.
  Here, $x$ is a string variable and input strings ending in the final state $s_0$ output variable $x$ (as shown by the label on the outgoing arrow from $s_0$.)
  }
  \label{fig:reverse}
  \vspace{-1.5em}
\end{wrapfigure}

Alur and {\v C}ern\'y \cite{AlurCerny10} proposed a one-way machine capable of realizing the same transformations as \dmsot{}s.
These machines, known as \emph{streaming string transducers} (\sst{}), work by storing and combining partial outputs in a finite set of variables, and enjoy a number of appealing properties including decidability of functional equivalence and type-checking (see \Cref{fig:reverse} for an \sst{} realization of \textsf{reverse}).
Alur and Deshmukh followed up this work by introducing nondeterministic streaming string transducers (\nsst{}) as a natural generalization \cite{AlurDeshmukh11} and proved this model captures precisely the same set of relations as \nmsot{}s.
Since the connection between automata and logic is often used as a yardstick for regularity, \mso{}-definable functions and relations over finite strings are often called regular functions and regular relations.

\noindent\textbf{Regular Relations over Infinite Strings.} The
expressiveness of \sst{}s and \mso{}-definable transformations also
coincide when representing functions over infinite strings \cite{AlurFiliotTrivedi12}.
Deterministic \sst{}s operating on infinite strings are known as \wdsst{}s, however, for regular relations of infinite strings, no existing computational model exists.
We combine and generalize results in the literature on \nsst{}s and \wdsst{}s to propose the computational model of nondeterministic streaming $\omega$-string transducers (\wnsst{}) that capture
regular relations of infinite strings.

\begin{figure}
% \scalebox{0.7}{
  \begin{tikzpicture}[node distance = 4cm, thick, >= stealth]
    \node[state,initial,accepting] (1) {$s_0$};
    \node[state,accepting]         (2)[right = 6cm of 1] {$s_1$};
    \node[state,accepting]         (3)[right = 3cm of 2] {$s_2$};
    
    \path[->] (1) edge[loop above]
    node[above] {$\alpha       \begin{cases}
        w := w\alpha
      \end{cases}$
} (1);
    \path[->] (1) edge[loop below]
    node[below]
    {
      $\#
      \begin{cases}
        w := \varepsilon \\
        z := z \#
      \end{cases}$
    } (1);
    \path[->] (1) edge node[below,pos=0.7]
         {$\alpha
           \begin{cases}
            w := w\alpha \\
             x := x\alpha \\
             y := \alpha y
           \end{cases}$
         } (2);
  \path[->] (2) edge[loop above]
  node[above]
  {$\alpha
           \begin{cases}
            w := w\alpha \\
             x := x\alpha \\
             y := \alpha y
           \end{cases}$
           }   (2);
    \path[->] (2) edge[bend right] node[above]
         {$\#
           \begin{cases}
            w := \varepsilon \\
             x := \varepsilon \\
             y := \varepsilon \\
             z := zyx\#
           \end{cases}$
         } (1);
    \path[->] (2) edge node[above]
         {$\alpha
           \begin{cases}
             w := \varepsilon \\
             z := zw\alpha
           \end{cases}$
         } (3);
    \path[->] (1) edge[bend right=35] node[below=0.35cm, pos=0.99]
         {$\alpha
           \begin{cases}
              w := \varepsilon \\
              z := zw\alpha
           \end{cases}$
         } (3);
  \path[->] (3) edge[loop above]
  node[above]
  {$\alpha
           \begin{cases}
             z := z\alpha
           \end{cases}$
           }   (3);
  \end{tikzpicture}
  % }
  \caption{\footnotesize An $\wnsst{}$ implementing the relation $R_{\protect\overleftarrow{u}u}$ from \Cref{ex:first}. Let $\alpha$ denote all symbols in $A$, excluding \#. Variable $w$ remembers the string since the last \#, while $x$ and $y$ store the chosen suffix and its reverse. The output variable is $z$.}
  \label{fig:wnsst}
  \vspace{-2em}
\end{figure}

\begin{example}
  \label{ex:first}
  Let $A$ be a finite alphabet and $\#$ be a special separator not in $A$.
  For $u, v \in A^*$, we say that $v \preceq u$ if $v$ is a suffix of $u$.
  Consider a relation $R_{\overleftarrow{u}u}$ that transforms strings in $(A \cup \set{\#})^\omega$ such that each maximal $\#$-free finite substring $u$ occurring in the input string is transformed into $\overleftarrow{v}v$ for some suffix $v$ of $u$.
  Formally, $R_{\overleftarrow{u}u}$ is defined as 
  \begin{multline*}
    \set{ 
      \left(u_1\#\cdots\#u_n\#w, 
      \overleftarrow{v_1}v_1\#\cdots\#\overleftarrow{v_n}v_n\#w \right)
      : u_i, v_i \in A^*, w \in A^\omega, \text{ and } v_i \preceq  u_i}
    \\
    \cup 
    \set{ 
      \left(u_1\#u_2\#\ldots, 
      \overleftarrow{v_1}v_1\#\overleftarrow{v_2}v_2\#\ldots\right)
      : u_i, v_i \in A^* \text{ and } v_i \preceq  u_i},
  \end{multline*}
  and can be implemented as an $\wnsst{}$ with B\"uchi acceptance condition (accepting states are visited infinitely often for accepting strings) as shown in \Cref{fig:wnsst}.
\end{example}

\noindent\textbf{Contributions and Outline.}
In \Cref{sec:wnsst} we introduce \wnsst{}s and their semantics as a computational model for regular relations.
In \Cref{sec:mso_nsst_eq} we prove that the \wnsst{}-definable relations coincide exactly with \mso{}-definable relations of infinite strings.
In \Cref{sec:dec} we consider regular model checking with regular relations. 
To enable regular model checking with regular relations, we study the following key verification problem.
The \emph{type checking problem} for \wnsst{}s asks to decide, given
two \wreg{} languages $L_1,L_2$ and an \wnsst{}, whether $\inter{T}(L_1) \subseteq L_2$, where $\inter{T}$ is the regular relation implemented by $T$.
We show that type checking for \wnsst{}s is decidable in \pspace{}.

%% file: wnsst.tex
An alphabet $A$ is a finite set of letters.
A \emph{string} $w$ over an alphabet $A$ is a finite sequence of symbols in $A$.
We denote the empty string by $\varepsilon$.
We write $A^*$ for the set of all finite strings over $A$, and for $w \in A^*$ we write $|w|$ for its length. 
A language $L$ over $A$ is a subset of $A^*$.
An \emph{\wstring{}} $x$ over $A$ is a function $x: \nats {\to} A$, and written as $x = x(0) x(1) \cdots$. % countably infinite sequence of letters in $A$.
We write $A^\omega$ for the set of all \wstring{}s over $A$, and $A^\infty$ for $A^* \cup A^\omega$.
An \wlang{} $L$ over $A$ is a subset of $A^\omega$.

\subsection{MSO Definable Relations}
\label{subsec:transducers}
Strings may be viewed as ordered structures encoded over the signature $\S_A = \{(a)_{a \in A}, <\}$ and interpreted with respect to $A^*$ or $A^\omega$.
The domain of a string in this context refers to the set of valid positions in the string, and the relation $<$ in $\S_A$ ranges over this domain.
The expression $a(x)$ holds true if the symbol at position $x$ is $a$, and $x < y$ holds if $x$ is a lesser index than $y$.

Formulae in \mso{} over $\S_A$ are defined relative to a
countable set of first-order variables $x, y, z, \ldots$ that range over individual
elements of the domain and a countable set of second-order variables $X, Y,
Z, \ldots$ that range over subsets of the domain. 
The syntax for well-formed formulae is given as:
\begin{equation*}
  \label{eq:mso-syntax}
  \phi ::= \exists X.\ \phi \mid \exists x.\ \phi \mid \phi \wedge \phi \mid \phi \vee \phi \mid \neg \phi \mid a(x) \mid x < y \mid x \in X
\end{equation*}

\mso{} transducers are particular specifications in this logic that define transformations between strings.
Intuitively, each such transducer copies each input string some fixed number of times and treats the positions in each copy as nodes in a graph, which are then relabeled and and rearranged in accordance with the formulae of the transducer to produce an output.

\begin{definition}
  A deterministic \mso{} \wstring{} transducer (\wdmsot{}) is a tuple
\begin{equation*}
  \left( A, B, \domain, N, (\phi^n_b(x))^{n \in N}_{b \in B}, (\psi^{n,m}(x,y))^{n,m \in N} \right),
\end{equation*}
where $A$ and $B$ are input and output alphabets, $N = \set{1,\ldots,n}$ is a set of copy indices, $\domain$ is an \mso{} sentence that defines an input language, the \emph{node formulae} $\left( \phi^n_b(x) \right)^{n \in N}_{b \in B}$ specify the labels of positions in the output, and the \emph{edge formulae} $(\psi^{n,m}(x,y))^{n,m \in N}$ specify which positions in the output will be adjacent.
\end{definition}

A \wdmsot{} operates over $N$ disjoint copies of the string graph of an input.
Each formula $\phi^n_b$ has a single free variable and should be interpreted such that if a position satisfies $\phi^n_b$, then that position will be
labeled by the symbol $b$ in the $n^{th}$ disjoint string graph comprising
the output. 
Each formula $\psi^{(n, m)}$ has two free variables and a satisfying pair of
indices indicates that there is a link between the former index in copy $n$ and
the latter index in copy $m$.

Nondeterminism is introduced through additional set variables $X_1, \dots, X_k$ called \emph{parameters}.
Fixing a valuation---sets of positions of the input graph satisfying the domain formula---of these parameters determines an output graph, just as in the deterministic case.
Each possible valuation may result in a different output graph for the same input graph, and thus nondeterminism arises from the choice of valuation.

\begin{definition}
A nondeterministic \mso{} \wstring{} transducer (\wnmsot{}) with $k$ free set variables $\vec{X}_k = (X_1,\ldots,X_k)$ is given as a tuple  
\begin{equation*}
  \left( A, B, \domain(\vec{X}_k), N, (\phi^n_b(x, \vec{X}_k))^{n \in N}_{b \in B}, (\psi^{n,m}(x, y, \vec{X}_k))^{n,m \in N} \right),
\end{equation*}
where all formulae are parameterized by the free second-order variables in addition to the required first-order parameters.
\end{definition}

A relation between strings is a \emph{regular relation} if it is definable by a \wnmsot{}.
Since \wdmsot{}s can map each input to at most one output, the relations definable by \wdmsot{}s are called the \emph{regular functions}.

\begin{example}
\label{ex:msot}
We now describe a \wnmsot{} capturing the relation given in \Cref{ex:first}.
Set $A = \set{a,b,\#} = B$, $N = \set{1, 2}$, and consider a single parameter $\vec{X}_1 = \set{X_1}$.
The domain of the relation is simply $A^\omega$, so we omit the formula.
For all symbols $\beta \in B$ and copy indices $n \in N$, the node formulae labels each position with the same symbol as the corresponding position in the input string: $\phi^n_\beta(x, X_1) \rmdef \beta(x)$.
In the interest of space, we defer formal specifications of the edge formulae to the appendix (cf. \Cref{sec:examples}) and describe the edge formulae informally. 
The formula for edges from copy 1 to copy 1 connects adjacent non-$\#$ positions that belong to $X_1$ in the reverse order.
The formula for edges from copy 1 to copy 2 connects non-$\#$ positions to themselves when the predecessor position is not in $X_1$.
The formula for edges from copy 2 to copy 2 links the right-most sequence of positions in $X_1$ that preceed a $\#$ symbol and also connect all those positions coming after the final $\#$ if required.
Finally, the formula for edges from copy 2 to copy 1 links $\#$ symbols to the last position in $X_1$ occurring left of the subsequent $\#$. 

Two possible outputs from the relation of \Cref{ex:first} are displayed in \Cref{fig:mso} which shows how the above \wnmsot{} constructs an output string for two different valuations of $X_1$.
A 1 in the blue (resp. green) row signifies that the position at that column is an element of $X_1$, while a 0 indicates that it is not an element of $X_1$.
\end{example}

\begin{figure}[t]
  % \vspace{-1.5em}
  \resizebox{\textwidth}{!}{%
   \centering
    \tikzstyle{trans}=[-latex, rounded corners]
    \begin{tikzpicture}[->,>=stealth',shorten >=1pt,auto,node distance=2cm,thick]
      \node (B) {$a$} ;
      \node at (2,0) (B0) {$b$} ;
      \node at (4,0) (B1) {$b$} ;
      \node at (6,0) (B2) {$b$} ;
      \node at (8,0) (B3) {$\#$} ;
      \node at (10,0) (B4) {$b$} ;
      \node at (12,0) (B5) {$a$} ;
      \node at (14,0) (B6) {$\#$} ;
      \node at (16,0) (B7) {$a^\omega$} ;
      
      \draw[trans] (B) -- (B0);
      \draw[trans] (B0) -- (B1); 
      \draw[trans] (B1) -- (B2); 
      \draw[trans] (B2) -- (B3); 
      \draw[trans] (B3) -- (B4); 
      \draw[trans] (B4) -- (B5); 
      \draw[trans] (B5) -- (B6); 
      \draw[trans] (B6) -- (B7); 
      
      \node at (0, -1) (C) {$a$} ;
      \node at (2,-1) (C0) {$b$} ;
      \node at (4,-1) (C1) {$b$} ;
      \node at (6,-1) (C2) {$b$} ;
      \node at (8,-1) (C3) {$\#$} ;
      \node at (10,-1) (C4) {$b$} ;
      \node at (12,-1) (C5) {$a$} ;
      \node at (14,-1) (C6) {$\#$} ;
      \node at (16,-1) (C7) {$a^\omega$} ;
      
      \draw[trans,color=blue] (C1) -- (C0); 
      \draw[trans,color=blue] (C2) -- (C1); 
      \draw[trans,color=blue] (C5) -- (C4); 
      
      \node at (0, -2) (D) {$a$} ;
      \node at (2,-2) (D0) {$b$} ;
      \node at (4,-2) (D1) {$b$} ;
      \node at (6,-2) (D2) {$b$} ;
      \node at (8,-2) (D3) {$\#$} ;
      \node at (10,-2) (D4) {$b$} ;
      \node at (12,-2) (D5) {$a$} ;
      \node at (14,-2) (D6) {$\#$} ;
      \node at (16,-2) (D7) {$a^\omega$} ;
      
      \draw[trans,color=blue] (C0) -- (D0);
      \draw[trans,color=blue] (D0) -- (D1); 
      \draw[trans,color=blue] (D1) -- (D2); 
      \draw[trans,color=blue] (D2) -- (D3); 
      \path[->] (D3) edge[trans,color=blue, bend right] (C5); 
      \draw[trans,color=blue] (C4) -- (D4); 
      \draw[trans,color=blue] (D4) -- (D5); 
      \draw[trans,color=blue] (D5) -- (D6); 
      \draw[trans,color=blue] (D6) -- (D7);

      \path[->] (C2) edge[trans,color=black!25!green, bend right] (C1); 

      \draw[trans,color=black!25!green] (C1) -- (D1);
      \path[->] (D1) edge[trans,color=black!25!green, bend right] (D2); 
      \path[->] (D2) edge[trans,color=black!25!green, bend right] (D3); 
      \path[->] (D3) edge[trans,color=black!25!green, bend left] (C5); 
      \path[->] (C5) edge[trans,color=black!25!green] (D5); 
      \path[->] (D5) edge[trans,color=black!25!green, bend right] (D6); 
      \path[->] (D6) edge[trans,color=black!25!green, bend right] (D7);
      
      {\color{blue} \node at (0, -3) (E) {0} ;
      \node at (2,-3) (E0) {1} ;
      \node at (4,-3) (E1) {1} ;
      \node at (6,-3) (E2) {1} ;
      \node at (8,-3) (E3) {0} ;
      \node at (10,-3) (E4) {1} ;
      \node at (12,-3) (E5) {1} ;
      \node at (14,-3) (E6) {0} ;
      \node at (16,-3) (E7) {} ;}

      {\color{black!25!green} \node at (0, -4) (F) {0} ;
      \node at (2,-4) (F0) {0} ;
      \node at (4,-4) (F1) {1} ;
      \node at (6,-4) (F2) {1} ;
      \node at (8,-4) (F3) {0} ;
      \node at (10,-4) (F4) {0} ;
      \node at (12,-4) (F5) {1} ;
      \node at (14,-4) (F6) {0} ;
      \node at (16,-4) (F7) {} ;}
    \end{tikzpicture}
  }
  \vspace{-2em}
  \caption{Two possible outputs of the relation given in \Cref{ex:first} constructed according ot the \wnmsot{} from \Cref{ex:msot}.}
  \label{fig:mso}
  % \vspace{-1.5em}
\end{figure}  

\subsection{Nondeterministic Streaming String Transducers}
\label{subsec:wnsst}
\begin{definition}
  A nondeterministic streaming string transducer $T$ over \wstring{}s (\wnsst{}) is a tuple $(A, B, S, I, \acc, \Delta, f, X, U)$, where 
\begin{itemize}
  \item $A$ and $B$ are finite input and output alphabets, 
  \item $S$ is a finite set of states, 
  \item $I \subseteq Q$ is a set of initial states, 
  \item $\acc$ is an acceptance condition, 
  \item $X$ is a finite set of string variables, 
  \item $U$ is a finite set of variable update functions of type $X \to (X \cup B)^*$, 
  \item $\Delta$ is a transition function of type $(S \times A) \to 2^{U \times S}$, and 
  \item $f \in X$ is an append-only output variable.
\end{itemize}
Such a machine is deterministic (a \wdsst{}) if $|\Delta(s,a)| =  1$, for all states $s \in S$ and symbols $a \in A$, and $|I| = 1$; it is nondeterministic otherwise.
\end{definition}

On each transition $s_k \xrightarrow[u_k]{a_k} s_{k+1}$, the transducer
changes state and applies the update $u_k$ to each variable of $X$ in
parallel. 
An \wnsst{} is \emph{copyless} if every variable in $X$ occurs at most once in the image $\im{u}$ of every update $u \in U$.
Alternately stated, an update $u \in U$ is copyless if the string $u(x_0)u(x_1) \ldots u(x_{n-1})$ has at most one occurrence of each $x \in X$, and an \wnsst{} is copyless if all of its updates are copyless.

A run of an \wnsst{} on an infinite string $a_1a_2\cdots \in A^\omega$ is an infinite sequence of states and
transitions $s_0 \xrightarrow[u_0]{a_0}
s_1 \xrightarrow[u_1]{a_1} \ldots$ where $s_0 \in I$ and $(s_{k+1}, u_k) \in \Delta(s_k, a_k)$ for all $k \in \nats$. 
Let $\runs{T}{w}$ be the set of all runs in $T$, given input $w$.
An update function $u : X \to (X \cup B)^*$ can easily be extended to $\widehat{u} : (X \cup B)^* \to (X \cup B)^*$ such that $\widehat{u}(w) \rmdef{} \varepsilon$ if $w = \varepsilon$, $\widehat{u}(w) \rmdef{} b\widehat{u}(w')$ if $w = bw'$, and $u(x)\widehat{u}(w')$ if $w = xw'$.
The effect of two updates $u_1, u_2 \in U$ in sequence can be \emph{summarized} by the function composition $\widehat{u}_1 \circ \widehat{u}_2$; likewise a sequence of updates of arbitrary length would be summarized by $\widehat{u}_0 \circ \widehat{u}_1 \circ \ldots \circ \widehat{u}_{n-1}$.
For notational convenience, we often omit the hats when the extension is clear from context.
Notice that if all updates in a sequence of compositions are copyless, then so is the entire summary.

A valuation is a function $X \to B^*$ mapping each variable to a string
value.
The initial valuation $\val{\varepsilon}$ of all variables is the empty
string $\varepsilon$. 
A valuation is well-defined after any finite prefix $r_n$ of a run $r$ and is computed as a composition of updates occurring on this prefix: $\val{r_n} = \val{\epsilon} \circ u_0 \circ u_1 \circ \cdots \circ u_{n-1}$.
The output $T(r) \rmdef \lim_{n \to \infty} \val{r_n}(f)$ of $T$ on $r$
is well-defined only if $r$ is accepted by $T$. 
Since the output variable $f$ is only ever appended to and never prepended, this limit exists and is an \wstring{} whenever $r$ is accepted, otherwise we set $T(r) = \bot$. 
The relation $\inter{T}$ realized by an \wnsst{} $T$ is given by $\inter{T} \rmdef \set{(w, T(r)) : r \in \runs{T}{w}}$.
An \wnsst{} $T$ is \emph{functional} if for every $w$ the set $\set{w' : (w,w') \in \inter{T}}$ has cardinality at most 1.

We consider both B\"uchi and Muller acceptance conditions for \wnsst{}s and reference these classes of machines by the initialisms \nbt{} and \nmt{} (\dbt{} and \dmt{} for their deterministic versions), respectively.
For a run $r \in \runs{T}{w}$, let $\Inf{r} \subseteq S$ denote the set of states visited infinitely often.
\begin{enumerate}
  \item A \emph{B\"uchi acceptance condition} is given by a set of states $F \subseteq S$ and is interpreted such that a \nbt{} is defined on an input $w \in A^\omega$ if there exists a run $r \in \runs{T}{w}$ for which $\Inf{r} \cap F \neq \emptyset$.
\item A \emph{Muller acceptance condition} is given as a set of sets $\fF = \set{F_0,\ldots,F_n} \subseteq 2^S$, interpreted such that a \nmt{} is defined on input $w \in A^\omega$ if there exists a run $r \in \runs{T}{w}$ for which $\Inf{r} \in \fF$.
\end{enumerate}

% Let $U^*$ be the set of all update functions generated by composing an arbitrary but finite sequence of functions from $U$.
% For a copyless \sst{}, define the \emph{extended transition 
% function} as $\Delta^* {:} (S \times A^*) {\to} 2^{S \times U^*}$ such that $\Delta^*(s,w) {\rmdef{}} \set{ (s_n, u) : u = u_1 \circ \ldots \circ u_{n-1}}$ 
% for $s \xrightarrow[u_0]{a_0} s_1 \xrightarrow[u_1]{a_1} \ldots \xrightarrow[u_{n-1}]{a_{n-1}} s_n \xrightarrow[u_n]{a_n} \ldots \in \runs{T}{w}$.

\begin{proposition}
  \label{thm:BMequal}
  A relation is \nbt{} definable if, and only if, it is \nmt{} definable.
\end{proposition}

The equivalence of \nbt{} and \nmt{} -definable relations follows from a straightforward application of the equivalence of nondeterministic B\"uchi automata and nondeterministic Muller automata.
Equivalence of these acceptance conditions in transducers allows us to switch between them whenever convenient.

\begin{remark} \label{remark:output}
Observe that \dmt{}s and functional \nmt{}s, both of which were introduced in~\cite{AlurFiliotTrivedi12}, have a slightly different output mechanism, which is defined as a function $\Omega : 2^S \rightharpoonup X^*$ such that the output string $\Omega(S')$ is copyless and of the form $x_1 \dots x_n$, for all $S' \subseteq S$ for which $\Omega(S') \neq \bot$.
Furthermore, there is the condition that if $s, s' \in S'$ and $a \in A$ s.t. $(u,s') \in \Delta(s, a)$, then 
(1) $u(x_k) = x_k$ for all $k < n$ and 
(2)  $u(x_n) = x_n w$ for some $w \in (X \cup B)^*$. 

In contrast, our definition has a unique append-only output variable $f \in X$.
However, our model with the Muller acceptance is as expressive as that studied in~\cite{AlurFiliotTrivedi12}.
One can use nondeterminism to guess a position in the input after which states in a Muller accepting set $S'$ will be visited infinitely often.
The output function can be defined by guessing a Muller set, and keeping an extra variable for the output.
Upon making the guess, it will move the contents of $x_1 \ldots x_n$ to the variable $f$ and make a transition to a copy $T_{S'}$ of the transducer where $\acc = \set{S'}$.
If any state outside the set $S'$ is visited, or the variables $x_1\ldots,x_{n-1}$ are updated, or the variable $f$ is assigned in non-appending fashion, then $T_{S'}$ makes a transition to a rejecting sink state. Alur, Filiot, and Trivedi~\cite{AlurFiliotTrivedi12} showed the equivalence of functional \nmt{} with \dmt{}.
This implies that the transductions definable using {\it functional} \nmt{}s or {\it functional} \nbt{}s (in our definition) are precisely those definable by \wdmsot{}. 
\end{remark}

\section{Equivalence of \wnmsot{} and \wnsst{}}
\label{sec:mso_nsst_eq}
Alur and Deshmukh \cite{AlurDeshmukh11} showed that relations over finite strings definable by nondeterministic \mso{} transducers coincide with those definable by nondeterministic streaming string transducers.
We generalize this result by proving that a relation is definable by an \wnmsot{} if, and only if, it is definable by an \wnsst{}.
We provide symmetric arguments to connect \wnsst{}, \wdsst{} and \wnmsot{}, \wdmsot{}, resulting in a simple proof.

Our arguments use the concept of a relabeling relation, following Engelfriet and Hoogeboom \cite{EngelfrietHoogeboom01}.
A relation $\rho \subseteq A^\omega \times B^\omega$ is a \emph{relabeling}, if there exists another relation $\rho' \subseteq A \times B$ such that $(aw, bv) \in \rho$ iff $(a,b) \in \rho'$ and $(w,v) \in \rho$. 
In other words, $\rho$ is obtained by lifting the  letter-to-letter relation $\rho'$, in a straight-forward manner, to \wstring{}s. Let 
$\letter(\rho)$ denote the letter to letter relation $\rho' \subseteq A \times B $ corresponding to $\rho$ 
and let \rl{} be the set of all such relabelings.

\begin{theorem}
  \label{thm:mso_nsst_eq}
  $\wnmsot{} = \wnsst{}$.
\end{theorem}

The proof of \Cref{thm:mso_nsst_eq} proceeds in two stages.
In the first part (\Cref{lemma:nsst-relab}), we show that every \wnsst{} is equivalent to the composition of a nondeterministic relabeling and a \wdsst{}.
In the second part (\Cref{lemma:msot-relab}), we show that every \wnmsot{} is equivalent to the composition of a nondeterministic relabeling and a \wdmsot{}.
These two lemmas, in conjunction with the equivalence of \dmt{}s and functional \nmt{}s~\cite{AlurFiliotTrivedi12}, allow us to equate these two models of transformation via a simple assignment.

\begin{lemma}
\label{lemma:nsst-relab}
$\wnsst{} = \wdsst{} \circ \rl{}$
\end{lemma}
\begin{proof}
We first show $\wdsst{} \circ \rl{} \subseteq \wnsst{}$ by proving that for every \dmt{} $T \rmdef (B,C,S,I,\fF,\Delta,f,X,U)$ and nondeterministic relabeling $\rho \subseteq A^\omega \times B^\omega$, there is a \nmt{} $T' \rmdef (A,C,S,I,\fF,\Delta_\rho,f,X,U)$ such that $\sem{T'} = \sem{T} \circ \rho$.
As indicated by the tuple given to specify $T'$, the only distinct components between the two machines are their input alphabets and their transition functions $\Delta$ and $\Delta_\rho$.
The latter is given as $\Delta_\rho \rmdef{} (s, a) \mapsto \bigcup\limits_{(a, b) \in \letter(\rho)} \Delta(s, b)$.
The nondeterminism of $\rho$ is therefore captured in $\Delta_\rho$.
This results in a unique run through $T'$, for every possible relabeling of inputs for $T$.
Since the remaining pieces of $T$ are untouched in the process of constructing $T'$, it is clear that $\sem{T'} = \sem{T} \circ \rho$.

What remains to be shown is the inclusion $\wnsst{} \subseteq \wdsst{} \circ \rl{}$: for any \nmt{} $T \rmdef (A,B,S,I,\fF,\Delta,f,X,U)$, there exists a \dmt{} $T'$ and a nondeterministic relabeling $\rho$ such that $\sem{T} = \sem{T'} \circ \rho$.
From $T$, we can construct a nondeterministic, letter-to-letter relation $\rho' \subseteq A \times (U \times S)$ as follows:
$\rho' \rmdef{} \set{(a,(u,s')) : (u,s') \in \Delta(s,a)}$.
Now let $\rho \subseteq A^\omega \times (U \times S)^\omega$ be the extension of $\rho'$ as described previously.
The relation $\rho$ contains the set of all possible runs through $T$ for any possible input in $A^\omega$.

Next, we construct a \dmt{} $T' \rmdef (U \times S,B,S,I,\fF,\Delta_\rho,f,X,U)$ with transition function $\Delta_\rho \rmdef (s,(u,s')) \mapsto \set{(u,s') \::\: (u,s') \in \Delta(s,a) \text{ for some } a\in A}$. 
Consequently, $T'$ retains only the pairs in $\rho$ which correspond to valid runs $T$ and encodes them as \wstring{}s over the alphabet $S \times U$.
The \dmt{} $T'$ then simply follows the instructions encoded in its input and thereby simulates only legitimate runs through $T$.
Thus, we may conclude that $\sem{T} = \sem{T'} \circ \rho$.
\qed
\end{proof}

\begin{lemma}
\label{lemma:msot-relab}
$\wnmsot{} = \wdmsot \circ \rl{}$.
\end{lemma}
\begin{proof}
We begin by showing the inclusion $\wnmsot{} \subseteq \wdmsot \circ \rl{}$: for any \wnmsot{} $T$, there exists an \wdmsot{} $T'$ and a relabeling $\rho$ such that $\inter{T} = \inter{T'} \circ \rho$.
Nondeterministic choice in $T$ is determined by the choice of assignment to free variables in $\vec{X}_k$.
Alternatively, the job of facilitating nondeterminism can be placed upon a relabeling relation, thereby allowing us to remove the parameter variables. 
Define a letter-to-letter relation $\rho' \subseteq A \times (A \times \set{0,1}^k)$ as follows: $\rho' \rmdef \set{(a, (a,b)) : b \in \set{0,1}^k}$, and let the relabeling $\rho \subseteq A^\omega \times (A \times \set{0,1}^k)^\omega$ be its extension.
This relabeling essentially gives us a new alphabet such that each symbol from $A$ is tagged with encodings of its membership status for each set parameter from $\vec{X}_k$.
Now, we can construct an \wdmsot{} $T'$ that is identical to $T$, apart from two distinctions.
Firstly, $T'$ is deterministic (i.e. it has no free set variables), and every occurrence of a subformula $x \in X_i$ in $T$ is replaced by a subformula $\bigvee\limits_{\substack{b \in \set{0,1}^k \wedge b[i] = 1}} (a, b)(x)$ in $T'$.
As a result of this encoding, the equality $\inter{T} = \inter{T'} \circ \rho$ holds. 

The converse inclusion, $\wdmsot \circ \rl \subseteq \wnmsot$, is much simpler.
%Every relabeling in \rl{} is \wnmsot{} definable, since a relabeling is easily seen to be a nondeterministic, rational transformation. 
%Because the relabeling relation has many outputs for a single input and each such output is fed to a \wdmsot{}, the resulting transformation is both \wnmsot{}-definable and
%nondeterministic.
%Thus, we conclude that any composition of a nondeterministic relabeling and a \wdmsot{} is definable by a \wnmsot{} and that $\wmsot \circ \rl \subseteq \wnmsot$.
Every relabeling $\rho$ in \rl{} is \wnmsot{} definable: consider $\rho' = \letter(\rho) \subseteq A \times B$. 
The $\wnmsot$ specifying $\rho$ is similar to identity/copy, except that here we have that the output label is $b$ iff the input label is $a$ and $(a, b) \in \rho'$.
This can be implemented using second-order variables $X_{b}$ for all $b \in B$. 
Let $\vec{X}_B$ represent this set. 
Only a single copy is required to produce the output.
Node formulae are given by $\phi^1_b(x,\vec{X}_B) \rmdef \bigvee\limits_{a \in A} \bigvee\limits_{(a, b) \in \rho'} (a(x) \wedge x \in X_b)$, and the edge formulae by $\psi^{1,1}(x, y,\vec{X}_B) \rmdef x < y$.
%Node and edge formulae remains same as for copy function except that the output label is $b$ only if position $x \in X_b$, position $x$ is labeled by $a$ in input and $(a, b) \in \rho'$.
It is known that $\wnmsot{}$ are closed under composition~\cite{CourcelleEngelfriet12}.
Thus, we conclude that any composition of a nondeterministic relabeling and a \wdmsot{} is definable by a \wnmsot{} and that $\wmsot \circ \rl \subseteq \wnmsot$.
\qed
\end{proof}

In conjunction \Cref{lemma:nsst-relab,lemma:msot-relab} along with the results of \cite{AlurFiliotTrivedi12} allow us to write the following equation, thereby proving \Cref{thm:mso_nsst_eq}.
\begin{equation*}
  \wnmsot{} = \wdmsot{} \circ \rl{} = \dmt{} \circ \rl{} = \nmt{} = \wnsst{}
\end{equation*}

%% file: decision.tex
In this section, we explain how algorithms for deciding properties of regular relations can be used to perform regular model checking.
Given two relations $T_1$ and $T_2$, their \emph{sequential composition} is $\sem{T_2 \circ T_1} \rmdef \set{(x,z) : (x,y) \in \sem{T_1}, (y,z) \in \sem{T_2}}$.
Let $T^k$ denote the $k$-fold composition of a relation $T$ with itself.
Let $T^*$ denote the transitive closure of $T$.

Suppose that $\init$ and $\bad$ are regular languages representing sets of states in some system that are initial, and unsafe, respectively.
Given a generic transition relation $T$ which captures the dynamics of the system, the \emph{regular model checking problem} asks to decide whether any element of $\bad$ is reachable from any element of $\init$ via repeated applications of $T$.
In precise terms, the regular model checking problem asks to decide whether the equation $\sem{T^*}(\init) \cap \bad = \emptyset$ holds.
Bounded model checking, in this setting, asks to decide, given $n \in \nats$, whether $\sem{T^k}(\init) \cap \bad = \emptyset$ holds, for all $k \leq n$.
Unbounded model checking is undecidable (cf. \Cref{sec:undec_proof} for a proof), even when $T$ is rational, so we focus on bounded model checking.

When $T$ is a rational relation, its image is always a regular language, and this permits the approach of iteratively applying $T$ from $\init$ and checking whether this set intersects with $\bad$ by standard automata-theoretic methods.
If $T$ is a regular relation, its image may not be a regular language, and we must iteratively compute compositions of $T$ with itself and test whether these compositions enter the $\bad$ language.
To allow this, we establish decidability of the \emph{type checking problem} for \wnsst{}s: given two $\omega$-regular languages $L_1,L_2$ and an \wnsst{} $T$, decide if the inclusions $L_1 \subseteq \dom{T}$ and $\sem{T}(L_1) \subseteq L_2$ hold.

\begin{reptheorem}{typecheck}
\label{thm:typecheck}
The type checking problem for \wnsst{}s is decidable in \pspace{}.
\end{reptheorem}

\input{type-checking-new}

Since regular relations are definable in \mso{}, they are closed under sequential composition.
In combination with \Cref{thm:mso_nsst_eq,thm:typecheck}, this establishes the necessary conditions for bounded regular model checking with regular relations to be possible.
Thus, we have the following corollary.

\begin{corollary}
Bounded model checking with regular relations is decidable.
\end{corollary}

Despite the fact that unbounded regular model checking is undecidable, bounded regular model checking provides a refutation procedure.
That is, it allows us to search for a witness for proving the system unsafe.
Unfortunately, we cannot use bounded model checking of this kind to decide if the system does satisfy the desired property.
On the other hand, we identify several special cases of the problem which permit the safety of the system to be verified in finite time.
In general, we assume that $\init \subseteq \overline{\bad}$, where $\overline{\bad}$ is the complement of $\bad$.

\paragraph{Functional Fixed Points.}
The first instance applies when $T$ is functional, i.e. $\sem{T}$ is a function, and relies on the following result of Alur, Filiot, and Trivedi \cite{AlurFiliotTrivedi12}.
\begin{theorem}
	\label{thm:AlurFiliotTrivedi12}
	Given an \wnsst{} $T$, it is decidable if $\sem{T}$ is a function.
	Given a pair of functional \wnsst{}s $T_1$ and $T_2$, it is decidable if $\sem{T_1} = \sem{T_2}$.
\end{theorem}
At every step of the bounded regular model checking procedure, one can check if $T^k$ is functional, if $T^{k+1}$ is functional, and if $\sem{T^k} = \sem{T^{k+1}}$.
If these three conditions hold, then, for all $m \geq 0$, we have that $\sem{T^k} = \sem{T^{k+m}}$.
When this occurs and $\sem{T^k}(\init) \subseteq \overline{\bad}$ holds, it follows that $\sem{T^k} = \sem{T^*}$ and therefore that $\sem{T^*}(\init) \subseteq \overline{\bad}$ which implies $\sem{T^*}(\init) \cap \bad = \emptyset$.
Note that $T^k$ can be functional even when $T$ is not.
To see this, consider a non-functional \wnsst{} $T$ such that $\sem{T}(a^\omega) = \set{b^\omega, c^\omega}$, and $\sem{T}(b^\omega) = d^\omega = \sem{T}(c^\omega)$.
If $a^\omega \in \init$ and $|\sem{T}(w)| = 1$ for every other input $w$ and $a^\omega \notin \im{T}$, then $T^2$ is functional.

\paragraph{Inductive Invariants.}
An alternative approach involves showing that $\sem{T}$ satisfies some inductive invariant.
Select, as a candidate invariant, a regular or \wreg{} language $L$ which is contained in the set of safe states $L \subseteq \overline{\bad}$.
Now, $L$ provides a witness to the unbounded safety of the system if the following pair of conditions are met: (i) $\init \subseteq L$ and (ii) $\sem{T}(L) \subseteq L$.
Together, (i) and (ii) imply that $\sem{T^*}(\init) \subseteq L$, and in combination with the assumption that $L \subseteq \overline{\bad}$ this yields that $\sem{T^*}(\init) \cap \bad = \emptyset$.
The necessary inclusions can be formulated as instances of the type checking problem, and so, given an appropriately chosen inductive invariant in the form of an \wreg{} language, the global safety of such a system may be verified in polynomial space.
This method is easily generalized by searching for $k$-inductive invariants: \wreg{} languages for which there is a $k \in \nats$ such that $\sem{T^k}(L) \subseteq L$.
The $k$-inductive approach complements bounded regular model checking, since, for a given $k$, bounded regular model checking lets us decide if the system is safe for up to $k$ transitions while $k$-induction lets us decide if it is safe after at least $k$ transitions.

%% file: type-checking-new.tex
%%%%%%%%%%%%%%

%Reviewer 2: In the proof of theorem 3.1 it must be added that every node and edge of the automata of exponential size can be queried in PSpace (if that is possible). 

%%%%%%%%%%%%%%

\begin{proof}
Suppose that $T \rmdef (A, B, S, I, F, \Delta, f, X, U)$ is an \nbt{} and $L_1 \subseteq A^\omega$ and $L_2 \subseteq B^\omega$ are \wreg{} languages, encoded, respectively, as deterministic Muller automata (\dma{}) $M_1$ and $M_2$.
We first check whether $ T$ is defined for all \wstring{}s $w \in L_1$,
i.e. whether $L_1 \subseteq \dom{T}$. 
A nondeterministic B\"uchi automaton (\nba{}) $\mathcal{C}$ that recognizes the domain of $T$ can be constructed in linear time by ignoring variables and output mechanism.
The inclusion $L_1 \subseteq \dom{T}$ can be decided in \pspace{} by checking emptiness of $M_1' \cap \mathcal{C}$ where $M_1'$ is the \nba{} equivalent to $M_1$ and $\mathcal{C}$ is the \nba{} representing the complement language of $\dom{T}$.
It is known that an \nba{} can be constructed from a \dma{} with exponential blowup in the number of states \cite{Boker18}. 
A complement automaton can be constructed for an \nba{} with exponential increase in the number of states as well \cite{Boker18}.
Hence $\mathcal{C}$ has exponentially many states relative to $T$ and $M_1$. 
Intersection of $M_1'$ and $\mathcal{C}$ is a standard product construction with a flag so that both $M_1'$ and $\mathcal{C}$ visit good states infinitely often. 
Thus the intersection \nba{} $M_1' \cap \mathcal{C}$ has exponentially many states relative to $T$ and $M_1$.
Thanks to the fact that emptiness of \nba{} can be checked in \nlogspace{}~\cite{Boker18}, the emptiness of this product automaton, can be decided in \npspace{} = \pspace{}.

We now assume that $T$ is well-defined on $L_1$ and construct a nondeterministic Muller automaton (\nma{}) $\Aa$ such that the language of $\Aa$ is defined as $\set{ w \in L_1 : \exists w' \in \sem{T}(w) \textnormal{ s.t. } w' \not\in L_2}$.
Next, we construct a \dma{} $\overline{M_2}$ for $\overline{L_2}$ by complementing the $\acc$ set.
%Suppose that its set of states is $S_{\overline{M_B}}$ and its accepting set is
%$\Ff_{\overline{M_B}}$.
%Similarly, let $S_{M_A}$ and $\Ff_{M_A}$ be the set of states and accepting
%sets of $\A$.
%Finally assume that $S$ and $F$ denote the set of states and accepting states of
%$T$. 
The automaton $\Aa$ simulates $M_1$, $T$ and $\overline{M_2}$ in
parallel.
Next, we construct an \nmt{} $T'$ corresponding to the \nbt{} $T$ in order to homogenize the acceptance condition accross these machines.
Let us fix the definition for all three machines: (i) $M_1 \rmdef (A, S_{1}, p_0, \Ff_{1}, \Delta_{1})$, (ii) $T' \rmdef (A, B, S, I, \Ff', \Delta, {\out}, X, U)$, (iii) $\overline{M_2} \rmdef (B, S_{2}, r_{0}, \Ff_2, \Delta_{2})$.

The \nma{} $\Aa$ is defined as the product of $M_1$ and $T'$ (without the output mechanism), and it stores a state summary map---i.e. the effect of running current valuation of each variable starting from all states of $\overline{M_2}$---in each of its own states. 
Formally, the states of $\Aa$ comprise a finite subset of $S_1 \times S \times \left(S_{2} \times X \to S_2 \cup \set{\bot}\right)$.
A state $(q, p, g)$ with $g(r, x) = r'$ represents that, starting from state $r$, if we read the current value of variable $x$, then we reach state $r'$.
If $g(r, x) = \bot$, it indicates that there is no run on valuation of $x$ starting from $r$.
This information can be updated along the run of $\Aa$.
For instance, if a transition of $T$ updates $x$ as $aybx$, then the summary map $g$ is updated to $g'$ such that $g'(r, x) = g(\Delta_2(g(\Delta_2(r, a), y), b), x)$, and summarizes the effect of reading $x = aybx$ in $\overline{M_2}$ starting from state $r$.

%Moreover, since it is not necessary that the product visits accepting states of $M_A$, $\overline{M_B}$ and $T$ simultaneously, we use an technique employed in the construction for intersection of \nba{}s.
%The idea is to have a circular layered construction where, after seeing an accepting state of one component, the automaton makes a transition to the next layer with going back to the first layer after seeing the last one.
The set of states of $\Aa$ is $S_{\Aa} = S_{1} \times S \times \left(S_2 \times X \rightarrow S_2 \cup \set{\perp}\right)$, in which $S_1$, $S$, and $S_2$ represent the state sets of $M_1$, $T'$, and $\overline{M_2}$, respectively.
The transition relation $\Delta_\Aa$ is defined such that $(q', p', g') \in \Delta_{\Aa}((q, p, g), a)$ iff 
(i) $\Delta_1(q, a) = q'$, 
(ii) $(u, p') \in \Delta_{1}(p, a)$, and
(iii) $g'(r, x) = r'$ and $\Delta_2(r, \val{u(x)}) = r'$, for all $x \in X$ and $r \in S_2$,.
Initial states are the product of initial states i.e. a set $I_{\Aa} = \{(q_0, p_0, r_0)  : q_0 \in I \}$.
The Muller accepting set of $\Aa$ is defined as the collection of all $P \subseteq S_{\Aa}$ such that (i) $\pi_1(P) \in \Ff_{1}$, (ii) $\pi_2(P) \in \Ff$, and (iii) $(\pi_3(P))(r_0, {\out}) \in \Ff_{2}$, where $\pi_i$ is the $i^{th}$ projection.
The size of $\nma{}$ $\Aa$ is exponential in the number variables of $T$, polynomial in the number of states of $M_1$ and $T$.
Thanks to the fact that emptiness of an \nma{} can be determined in \nlogspace{} \cite{Boker18}, emptiness of $\Aa$ having exponential states in the inputs $T$, $M_1$ and $M_2$, can be decided in
\npspace{} and thus, by Savitch's theorem, also in \pspace{}.
\qed
\end{proof}

% \begin{remark}
% 	Note that in the above proof, the complexity of the decision problem may differ if the input automata are given with different acceptance conditions. 
% \end{remark}

%% file: conclusion.tex
We introduced \wnsst{}s as a computational model for regular relations over infinite strings, and showed that the relations definable by \wnsst{} coincide exactly with those definable in \mso{}.
Motivated by potential applications in formal verification, we studied algorithmic properties of these objects and established the minimal theoretical results required for bounded regular model checking to be possible with regular transition relations.
Regular functions and relations provide an intriguing class of models for real valued functions, see \Cref{fig:sst-square} for example.
In \cite{ChaudhuriSankaranarayananVardi13,GormanHKMWWXY19} analytic properties such as continuity and differentiability of real functions encoded by $\omega$-automata have been studied.
Extending this line of research by going beyond standard $\omega$-automata is both theoretically interesting and could be leveraged towards applications involving verification and control of dynamical systems.
The present work indicates the viability of generalizing the automata-theoretic approach to modeling real functions.
With this application in mind, it would be worthwhile to study the approximation techniques developed for traditional regular model checking to see if they generalize to handle regular relations.

%% file: formulae.tex
For ease of notation, we define the following helper formulae 
\begin{align*}
  \mathsf{first}(x) &\rmdef{} \neg \exists y.\ y < x, \\
  \mathsf{prec}(x,y) &\rmdef{} x < y \wedge \neg (\exists z.\ z < y \wedge x < z), \\
  \mathsf{btw}(z,x,y) &\rmdef (x < z \wedge z < y) \vee (y < z \wedge z < x).
\end{align*}
We also abbreviate $\neg x \in X$ as $x \notin X$.
\begin{align*}
    \psi^{1,1}(x, y, X_1) \rmdef{} & x \in X_1 \wedge y \in X_1 \wedge \neg \#(x) \wedge \neg \#(y) \wedge \mathsf{prec}(y, x) \wedge \\
    & \exists z_1.\ \#(z_1) \wedge x < z_1 \wedge \neg (\exists z_2.\ \mathsf{btw}(z_2, x, z_1) \wedge z_2 \notin X_1) \\
		& \\
    \psi^{1,2}(x, y, X_1) \rmdef{} & \neg (x < y \vee y < x) \wedge x \in X_1 \wedge \\
    & (\mathsf{first}(x) \vee \exists z.\ \mathsf{prec}(z, x) \wedge z \notin X_1) \\
		& \\
    \psi^{2,2}(x, y, X_1) \rmdef{} & \mathsf{prec}(x, y) \wedge \Big( \big( \#(x) \wedge \neg (\exists z.\ x < z \wedge \#(z) \big) \vee \\
		& \big(\exists z_1.\ y < z_1 \wedge \#(z_1) \wedge \neg (\exists z_2.\ \mathsf{btw}(z_2, y, z_1) \wedge  z_2 \notin X_1) \\
    & \wedge x \in X_1 \wedge (y \in X_1 \vee \#(y)) \big) \Big) \\
		& \\
    \psi^{2,1}(x, y, X_1) \rmdef{} & \#(x) \wedge x < y \wedge \big( (\neg \#(y) \wedge y \in X_1 \wedge \exists z.\ \mathsf{prec}(y,z) \wedge \#(z)) \\
    & \vee (\#(y) \wedge \neg \exists z.\ \mathsf{btw}(z,x,y) \wedge z \in X_1) \big)
\end{align*}

%% file: undec.tex
We show undecidability of the regular model checking problem in this framework by giving a
reduction from the undecidable halting problem for two-counter machines. 
\begin{framed}
A \emph{two-counter machine} (Minsky machine) $\M$ is a tuple $(L, C)$
where:
${L = \set{\ell_1, \ell_2, \ldots, \ell_n}}$ is the set of
instructions and 
${C = \set{c_1, c_2}}$ is the set of two \emph{counters}. 
There is a distinguished terminal instruction  $\ell_n$ called
HALT and the instructions $L$ are one of the following types:
\begin{description}[leftmargin=1in]
  \item[increment] $\ell_i : c := c+1$;  goto  $\ell_k$,
  \item[decrement] $\ell_i : c := c-1$;  goto  $\ell_k$,
  \item[zero-test] $\ell_i$ : if $(c{>}0)$ then goto $\ell_k$ else goto $\ell_m$,
  \item[Halt] $\ell_n:$ HALT.
\end{description}
where $c \in C$, $\ell_i, \ell_k, \ell_m \in L$.
In the following, we replace $\ell_n$ by $\ell_{halt}$.
Let $I,D,O$ represent the set of increment, decrement and  
zero-check instructions s.t. $L = I \cup D \cup  O$. 
\end{framed}
 
A configuration of a two-counter machine is a tuple $(\ell, c, d)$ where
$\ell \in L$ is an instruction, and $c, d$ are natural numbers that specify the
value
of counters $c_1$ and $c_2$, respectively.
The initial configuration is $(\ell_1, 0, 0)$.
A run of a two-counter machine is a (finite or infinite) sequence of
configurations $\seq{k_1, k_2, \ldots}$ where $k_1$ is the initial
configuration, and the relation between subsequent configurations is
governed by transitions between respective instructions.
The run is a finite sequence if and only if the last configuration is
the terminal instruction $\ell_{halt}$.
Note that a two-counter  machine has exactly one run starting from the initial
configuration. We assume without loss of generality that $\ell_0$ is 
an increment instruction. Clearly, it cannot be a decrement instruction.
If $\ell_1$ were a zero check instruction, we add two dummy instructions $\ell'_1, \ell''_1$
such that $\ell'_1$
increments a counter, and $\ell''_1$ decrements it, and passes control to $\ell_1$.
The dummy instructions $\ell'_1, \ell''_1$ will never again be encountered in the two
counter machine after control passes to $\ell_1$.
 
The \emph{halting problem} for a two-counter machine asks whether 
its unique run ends at the terminal instruction $\ell_{halt}$.
It is well known that the halting problem for two-counter machines is
undecidable. 

Let our alphabet be $A = \{0,1\}$ and let $M = (L,C)$ be a two-counter machine.
A configuration $k_{i} = (\ell_{i},c,d)$ of $M$ can be represented as the string 
$w_{i} = 0^{i}10^{c}10^{d}$.
The intial configuration is then $011$.
For each instruction in $M$, we can construct an SST $T$ such that
$k_{i}\vdash k_{j}$ iff $\sem{T}(w_i) = w_{j}$. Then we can construct a larger transducer
that first reads the instruction portion of the configuration and branches to 
execute the appropriate sub-machine.
Thus, each instruction transducer reads the suffix coming after the first 1.
The instruction transducer for each type of instruction involving $c$ is shown below;
similar machines could be constructed for the instructions involving $d$.
\begin{description}
  \item[increment]
  $\ell_i : c := c+1$;  goto  $\ell_k$,\\
  \begin{figure}[h!]
  \begin{tikzpicture}[> = stealth, shorten > = 1pt,	auto, node distance = 3cm, semithick]
    \tikzstyle{every state}=[draw=black, thick,	fill = white,	minimum size = 4mm]
    \node[initial,state]   (0)             {0};
    \node[state,accepting] (1)[right of=0] {1};
    \node                  (F)[right of=1] {};
    \path[->] (0) edge[loop above] node {$0 \mid x:=x0$}            (0);
    \path[->] (0) edge             node {$1 \mid x:=x01$}           (1);
    \path[->] (1) edge[loop above] node {$a \mid x:=xa$}  (1);
    \path[->] (1) edge node {$0^{k}1x$} (F);
  \end{tikzpicture}
\end{figure}\\\\
\item[decrement]
  $\ell_i : c := c-1$;  goto  $\ell_k$,\\\\
  \begin{figure}[h!]
    \begin{tikzpicture}[> = stealth, shorten > = 1pt,	auto, node distance = 3cm, semithick]
      \tikzstyle{every state}=[draw=black, thick,	fill = white,	minimum size = 4mm]
      \node[initial,state]   (0)             {0};
      \node[state,accepting] (1)[right of=0] {1};
      \node                  (F)[right of=1] {};
      \path[->] (0) edge             node {$0 \mid x:=\varepsilon$}  (1);
      \path[->] (1) edge[loop above] node {$a \mid x:=xa$} (1);
      \path[->] (1) edge node {$0^{k}1x$} (F);
    \end{tikzpicture}
  \end{figure}\\
\item[zero-test]
  $\ell_i$ : if $(c{>}0)$ then goto $\ell_k$
  else goto $\ell_m$,\\
  \begin{figure}[h!]
    \begin{tikzpicture}[> = stealth, shorten > = 1pt,	auto, node distance = 3cm, semithick]
      \tikzstyle{every state}=[draw=black, thick,	fill = white,	minimum size = 4mm]
      \node[initial,state]   (0)                               {0};
      \node[state,accepting] (1)[above right=1cm and 3cm of 0] {1};
      \node[state,accepting] (2)[below right=1cm and 3cm of 0] {2};
      \node    (F)[right of = 1] {};
      \node    (G)[right of = 2] {};
      \path[->] (0) edge             node[sloped] {$1 \mid x:=x1$}           (1);
      \path[->] (0) edge             node[below,sloped] {$0 \mid x:=x0$}           (2);
      \path[->] (1) edge[loop above] node {$a \mid x:=xa$} (1);
      \path[->] (2) edge[loop above] node {$a \mid x:=xa$} (2);
      \path[->] (1) edge node {$0^{k}1x$} (F);
      \path[->] (2) edge node {$0^{m}1x$} (G);
    \end{tikzpicture}
  \end{figure}
\end{description}
The only remaining task is to decide when to execute each sub-machine that we have just defined.
This is achieved by creating one state for each instruction in $L$, and attaching an outgoing
transition on 1 to the $i$th instruction transducer from the $i$th state.
The complete transducer of the system reads in the prefix of 0s in a configuration and executes
the transducer associated with the instruction represented by that unary number.
For example, if $\ell_{i}$ is the instruction given above as decrement, a partial view of thes
complete transducer would look as depicted in Figure~\ref{fig:undec}.
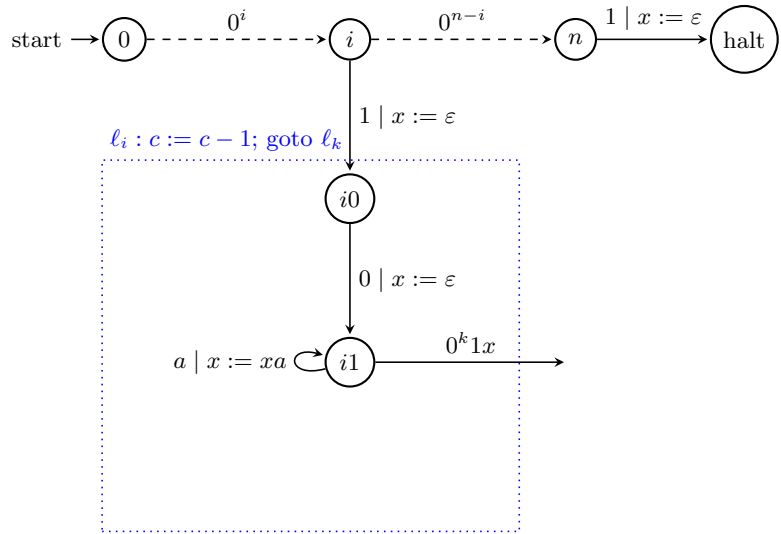
\begin{figure}[h!]
  \begin{tikzpicture}[> = stealth, shorten > = 1pt,	auto, node distance = 3cm, semithick]
    \tikzstyle{every state}=[draw=black, thick,	fill = white,	minimum size = 4mm]
    \node[initial,state] (0)                     {0};
    \node[state]         (i)[right of=0]         {$i$};
    \node[state]         (n)[right of=i]         {$n$};
    \node[state]         (i0)[below=1.5cm of i]  {$i0$};
    \node[state]         (i1)[below=1.5cm of i0] {$i1$};
    \node[state]         (h)[right=1.5cm of n]   {halt};
    \node                (F)[right of = i1]     {};
    \node (x)[below=1.5cm of i1] {};
    \node (y)[right=1.5cm of i0] {};
    \node (z)[below=2cm of 0] {};
    \node[inner sep=5pt,draw=blue,dotted,fit={(i0) (z) (i1) (x) (y)}] (box) {};
    \node[above right,text=blue] at (box.north west) {$\ell_i : c := c-1$;  goto  $\ell_k$};
    \path[->] (0)  edge[dashed]     node {$0^{i}$}             (i);
    \path[->] (i)  edge[dashed]     node {$0^{n-i}$}           (n);
    \path[->] (i)  edge             node {$1 \mid x:=\varepsilon$}  (i0);
    \path[->] (i0) edge             node {$0 \mid x:=\varepsilon$}  (i1);
    \path[->] (i1) edge[loop left] node {$a \mid x:=xa$} (i1);
    \path[->] (n)  edge             node {$1 \mid x:=\varepsilon$}  (h);
    \path[->] (i1) edge node {$0^{k}1x$} (F);
  \end{tikzpicture}
  \caption{}
  \label{fig:undec}
\end{figure}
It is clear that if there exists a way to decide whether a halting configuration ($0^{n}10^{*}10^{*}$) is reachable,
then there must exist a solution to the two-counter machine halting problem.